\documentclass[journal]{IEEEtranTIE}
\usepackage{graphicx}
\pdfoutput=1
\usepackage{cite}
\usepackage{picinpar}
\usepackage{amsmath}
\usepackage{url}
\usepackage{flushend}
\usepackage[latin1]{inputenc}
\usepackage{colortbl}
\usepackage{soul}
\usepackage{multirow}
\usepackage{pifont}
\usepackage{color}
\usepackage{alltt}
\usepackage[hidelinks]{hyperref}
\usepackage{enumerate}
\usepackage{siunitx}
\usepackage{breakurl}
\usepackage{epstopdf}
\usepackage{pbox}
\usepackage{subfigure}
\usepackage{balance}

\newtheorem{assumption}{Assumption}
\newtheorem{lemma}{Lemma}
\newtheorem{remark}{Remark}

\newtheorem{theorem}{Theorem}

\newenvironment{proof}{{\indent \indent \it Proof:}}

\begin{document}
\title{	Prescribed Performance Adaptive Fixed-Time Attitude Tracking Control of a 3-DOF Helicopter with Small Overshoot}

\author{
	\vskip 1em
	
	Xidong Wang

	\thanks{
					
		Xidong Wang is with the Research Institute of Intelligent Control and Systems, School of Astronautics, Harbin Institute of Technology, Harbin 150001, China (e-mail: 17b904039@stu.hit.edu.cn). 
	}
}

\maketitle
	
\begin{abstract}
In this article, a novel prescribed performance adaptive fixed-time backstepping control strategy is investigated for the attitude tracking of a 3-DOF helicopter. First, a new unified barrier function (UBF) is designed to convert the prescribed performance constrained system into an unconstrained one. Then, a fixed-time (FxT) backstepping control framework is established to achieve the attitude tracking. By virtual of a newly proposed inequality, a non-singular virtual control law is constructed. In addition, a FxT differentiator with a compensation mechanism is employed to overcome the matter of ``explosion of complexity". Moreover, a modified adaptive law is developed to approximate the upper bound of the disturbances. To obtain a less conservative and more accurate approximation of the settling time, an improved FxT stability theorem is proposed. Based on this theorem, it is proved that all signals of the system are FxT bounded, and the tracking error converges to a preset domain with small overshoot in a user-defined time. Finally, the feasibility and effectiveness of the presented control strategy are confirmed by numerical simulations.
\end{abstract}


{}

\definecolor{limegreen}{rgb}{0.2, 0.8, 0.2}
\definecolor{forestgreen}{rgb}{0.13, 0.55, 0.13}
\definecolor{greenhtml}{rgb}{0.0, 0.5, 0.0}

\section{Introduction}

\IEEEPARstart{I}{n} recent decades, as one of the typical rotorcrafts, the helicopter has been paid substantial attention because of its widespread applications in various scenarios \cite{1.Li2015F}. However, the high-performance tracking control design of the helicopter remains a challenging project due to the characteristics of the system \cite{2.li2021}. 

\begin{figure}[!t]\centering
	\includegraphics[width=8.5cm]{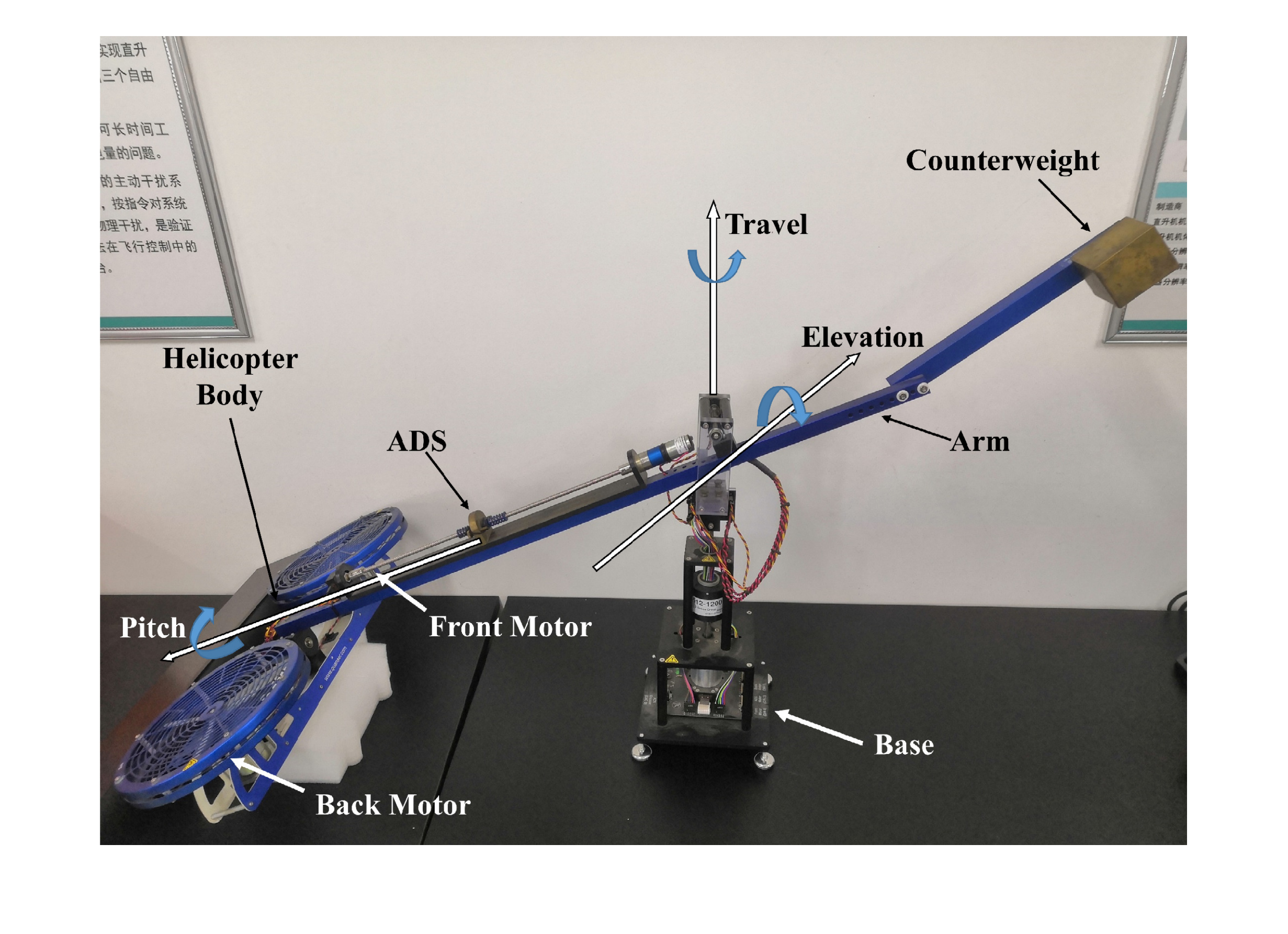}
	\caption{Structure of the 3-DOF helicopter system}
\end{figure}

Owing to the advantages of high tracking accuracy, fast response, and anti-interference capability, finite-time (FnT) control approaches have become a research hotspot \cite{fnt}. Recently, a variety of FnT control schemes have been developed and applied to the 3-DOF helicopter platform (Fig. 1). In \cite{fnt1}, a robust regulation controller for the laboratory helicopter was constructed based on the super-twisting sliding mode observer (STSMO), which was utilized to estimate the lumped disturbances and identify the state vector in finite time. To tackle the problem that the upper bound of disturbance is usually unknown in advance, the authors in \cite{fnt2} developed an adaptive-gain super-twisting algorithm for the tracking control of the helicopter platform. In \cite{fnt3}, an adaptive FnT control law based on a continuous differentiator was presented for the lab helicopter system, which can effectively suppress chattering. An adaptive smooth FnT control protocol, integrated with a singularity-free integral sliding mode surface, was established to fulfill the helicopter attitude tracking in \cite{2020arXiv}. In \cite{2021arXiv}, a novel control strategy, which embraced a modified FnT backstepping controller based on a newly designed disturbance observer, was presented for the tracking problem of the lab helicopter platform.

The drawback of the above strategies is that the settling time is related to the initial conditions, which is difficult to obtain accurately in many practical systems. Moreover, when the initial states are far from the equilibrium points, the settling time will become longer. To work out this problem, the FxT stability was first put forward in \cite{fxt}, where the settling time has no connection with the initial values. Thereafter, numerous FxT control approaches were developed to solve the control problems of different kinds of systems \cite{fxt1,fxt2,fxt3,fxt4}. 

Although the above-mentioned control laws can fulfill the satisfactory steady-state performance, the tracking errors' transient behavior is seldom taken into account. For the sake of addressing this problem, a prescribed performance control approach was first presented in \cite{PPC1}, which introduced the concepts of error transformation and performance function (PF) to guarantee the behavior at both transient and steady states. In \cite{PPC2}, the FnT prescribed tracking control was achieved by means of a newly designed PF. Authors in \cite{PPC3} proposed a novel prescribed performance FxT control scheme to realize the guaranteed tracking of strict-feedback nonlinear systems. By incorporating the PF with backstepping algorithm, a FxT adaptive control protocol was developed in \cite{PPC4} to cope with the trajectory tracking problem. In addition, a FxT differentiator with a novel auxiliary compensator was adopted to handle the ``explosion of complexity".

Enlightened by the aforementioned observations, this article develops a novel prescribed performance adaptive fixed-time control protocol for the attitude tracking of the 3-DOF helicopter system via the backstepping algorithm. First, a newly designed UBF is utilized to transform the prescribed performance constrained system into an unconstrained one. Then, motivated by \cite{PPC4}, a FxT backstepping control strategy is established to fulfill the attitude tracking. By using a newly proposed inequality, a non-singular virtual control law is constructed. In addition, a FxT differentiator with a compensator is adopted to overcome the matter of ``explosion of complexity". To compensate the impact of disturbances, a modified adaptive law is developed. The major innovations of this article can be stated as follows:
\begin{enumerate}[1)]
	\item A novel UBF is developed to convert the prescribed performance restrained system into an unrestrained one, which can effectively diminish the control input by selecting appropriate parameters.
	\item An improved FxT stability theorem is proposed to provide a less conservative and more accurate approximation of the settling time than the classical result.
\end{enumerate}

Based on the modified FxT stability theorem, it is proved that all signals of the system are FxT bounded, while the tracking error converges to a predetermined domain with small overshoot in a user-defined time. The feasibility and effectiveness of the proposed control protocol are confirmed by numerical simulations. 

The rest of this article is arranged as follows. The mathematical model of the helicopter system, an improved FxT stability theorem and some useful lemmas are presented in Section II. Section III describes the control design procedure. Numerical simulations are conducted in Section V to show the feasibility and effectiveness of the designed control protocol. Section VI draws the conclusion.

Notation: In this article, ${\mathop{\rm sig}\nolimits} {\left( x \right)^r } = {\left| x \right|^r }{\mathop{\rm sgn}} (x)$.
\section{Problem formulation and preliminaries}
\subsection{Problem Formulation}
In light of the attitude tracking problem of elevation and pitch angles, the mathematical model of the helicopter system can be expressed as \cite{2020arXiv}

\begin{equation}
\begin{aligned}
{{\dot x}_1} &= {x_2}\\
{{\dot x}_2} &= \frac{{{L_a}}}{{{J_\alpha }}}{{\bar u}_1} - \frac{g}{{{J_\alpha }}}m_e{L_a}\cos ({x_1}) + {d_1}(x,t)\\
{{\dot x}_3} &= {x_4}\\
{{\dot x}_4} &= \frac{{{L_h}}}{{{J_\beta }}}{{\bar u}_2} + {d_2}(x,t)
\end{aligned}
\end{equation}

The definitions of variables and the values of parameters can be found in \cite{2020arXiv}. 

The control aim is to design a prescribed performance FxT control protocol such that the tracking errors can converge to a predetermined domain with small overshoot in a user-defined time, while all signals of the system are bounded within a fixed time. 
\begin{assumption}
The desired tracking signals ${x_{1d}}(t),{x_{3d}}(t)$ and their derivatives are bounded, smooth and known.
\end{assumption}
\begin{assumption}
The disturbances are bounded as well as smooth.
\end{assumption}
\subsection{Improved FxT stability}
\begin{lemma}
Consider the following system 
\begin{equation}
\dot x(t) = f(x(t)),x(0)={x_0}
\end{equation}
where $x \in {R^n}$, $f:{R^n} \to {R^n}$ and $f(0) = 0$. If there exists a continuous Lyapunov function $V\left( x \right)$ satisfying:
\begin{equation}
\dot V(x(t)) \le  - {\mu _1}V{(x(t))^p} - {\mu _2}V{(x(t))^q}{\rm{ + }}{\mu _3}
\end{equation}
where ${\mu _1} > 0,{\mu _2} > 0,0 < {\mu _3} < \infty ,0 < p < 1,q > 1$, then the origin of system (2) is practical FxT stable, and the residual set is bounded by:
\begin{equation}
{F_x} = \left\{ {x:V\left( x \right) \le \min \left\{ {{{\left( {\frac{{{\mu _3}}}{{\left( {1 - \tau } \right){\mu _1}}}} \right)}^{\frac{1}{p}}},{{\left( {\frac{{{\mu _3}}}{{\left( {1 - \tau } \right){\mu _2}}}} \right)}^{\frac{1}{q}}}} \right\}} \right\}
\end{equation}
where $\tau  \in (0,1)$. Moreover, the settling time $T$ can be estimated by:
\begin{equation}
\begin{aligned}
T \le {T_1} = \max \left\{ {\frac{{\Gamma \left( {\frac{{1 - p}}{{q - p}}} \right)\Gamma \left( {\frac{{q - 1}}{{q - p}}} \right)}}{{{\mu _1}\left( {q - p} \right)}}{{\left( {\frac{{{\mu _1}}}{{\tau {\mu _2}}}} \right)}^{\frac{{1 - p}}{{q - p}}}}},\right.\\
\phantom{=\;\;}\left.{\frac{{\Gamma \left( {\frac{{1 - p}}{{q - p}}} \right)\Gamma \left( {\frac{{q - 1}}{{q - p}}} \right)}}{{\tau {\mu _1}\left( {q - p} \right)}}{{\left( {\frac{{\tau {\mu _1}}}{{{\mu _2}}}} \right)}^{\frac{{1 - p}}{{q - p}}}}} \right\}
\end{aligned}
\end{equation}
where $\Gamma \left( z \right) = \int_0^{{\rm{ + }}\infty } {\exp \left( { - \theta } \right)} {\theta ^{z - 1}}d\theta $ is the gamma function.
\end{lemma}

\begin{proof}
For a scalar $\tau  \in (0,1)$, (3) can be rewritten as
\begin{equation}
\dot V(x) \le  - \tau {\mu _1}V{(x)^p} - (1 - \tau ){\mu _1}V{(x)^p} - {\mu _2}V{(x)^q}{\rm{ + }}{\mu _3}
\end{equation}

or
\begin{equation}
\dot V(x) \le  - {\mu _1}V{(x)^p} - \tau {\mu _2}V{(x)^q} - (1 - \tau ){\mu _2}V{(x)^q}{\rm{ + }}{\mu _3}
\end{equation}

Then using \emph{Proposition 1} in \cite{fixed}, the approximation of the settling time (5) can be obtained.
\end{proof}

\begin{remark}
The classical result for the approximation of the settling time is \cite{fixed1,fixed2}
\begin{equation}
{T_2} = \frac{1}{{\tau {\mu _1}(1 - p)}} + \frac{1}{{\tau {\mu _2}(q - 1)}}
\end{equation}

It can be proved that $T_1 < T_2$, which means that \emph{Lemma 1} provides a less conservative and more accurate approximation of the settling time than the classical result.
\end{remark}
\begin{remark}
Denote $l = \frac{{1 - p}}{{q - p}}$. By means of $\Gamma \left( l \right)\Gamma \left( {1 - l} \right) = \frac{\pi }{{\sin \left( {l\pi } \right)}}$, $T_1$ can be rewritten as
\begin{equation}
\begin{aligned}
{T_1} = \max \left\{ {\frac{\pi }{{{\mu _1}\left( {q - p} \right)\sin \left( {l\pi } \right)}}{{\left( {\frac{{{\mu _1}}}{{\tau {\mu _2}}}} \right)}^l}},\right.\\
\phantom{=\;\;}\left.{\frac{\pi }{{\tau {\mu _1}\left( {q - p} \right)\sin \left( {l\pi } \right)}}{{\left( {\frac{{\tau {\mu _1}}}{{{\mu _2}}}} \right)}^l}} \right\}
\end{aligned}
\end{equation}
\end{remark}

If $p+q=2$, inspired by \cite{new1,new2}, a further approximation of the settling time $T$ is obtained as follows.
\begin{lemma}
If $p+q=2$, the settling time $T$ of \emph{Lemma 1} can be further estimated by:
\begin{equation}
\begin{aligned}
T \le \bar T \cdot \max \left\{ {\frac{\pi }{2} - \arctan \left( {\sqrt {\frac{{{\mu _2}}}{{\tau {\mu _1}}}} {{\left( {\frac{{{\mu _3}}}{{(1 - \tau ){\mu _1}}}} \right)}^{\frac{{1 - p}}{p}}}} \right)},\right.\\
\phantom{=\;\;}\left.{\frac{\pi }{2} - \arctan \left( {\sqrt {\frac{{\tau {\mu _2}}}{{{\mu _1}}}} {{\left( {\frac{{{\mu _3}}}{{(1 - \tau ){\mu _2}}}} \right)}^{\frac{{1 - p}}{{2 - p}}}}} \right)} \right\}
\end{aligned}
\end{equation}
where $\bar T = \frac{1}{{(1 - p)\sqrt {\tau {\mu _1}{\mu _2}} }}$.
\end{lemma}

The result of \emph{Lemma 2} can be extended to the following more general forms:
\begin{lemma}
If $(a - 1)p + q = a$ with $a \ge 2,a \in {\bf{N^ + }}$, the settling time $T$ is approximated  by $T \le \max \left\{ {{T_3},{T_4}} \right\}$, where
\begin{equation}\small
\begin{aligned}
{T_3} = {\bar T_3}  \left( {\frac{\pi }{{a\sin \left( {\frac{\pi }{a}} \right)}} - {I_a}\left( {{{\left( {\frac{{{\mu _2}}}{{\tau {\mu _1}}}} \right)}^{\frac{1}{a}}}{{\left( {\frac{{{\mu _3}}}{{(1 - \tau ){\mu _1}}}} \right)}^{\frac{{1 - p}}{p}}}} \right) + {I_a}(0)} \right)
\end{aligned}
\end{equation}
\begin{equation}\small
\begin{aligned}
{T_4} = {\bar T_4} \left( {\frac{\pi }{{a\sin \left( {\frac{\pi }{a}} \right)}} - {I_a}\left( {{{\left( {\frac{{\tau {\mu _2}}}{{{\mu _1}}}} \right)}^{\frac{1}{a}}}{{\left( {\frac{{{\mu _3}}}{{(1 - \tau ){\mu _2}}}} \right)}^{\frac{{1 - p}}{q}}}} \right) + {I_a}(0)} \right)
\end{aligned}
\end{equation}
with ${\bar T_3} = \frac{1}{{(1 - p)\tau {\mu _1}}}{\left( {\frac{{\tau {\mu _1}}}{{{\mu _2}}}} \right)^{\frac{1}{a}}}$, ${\bar T_4} = \frac{1}{{(1 - p){\mu _1}}}{\left( {\frac{{{\mu _1}}}{{\tau {\mu _2}}}} \right)^{\frac{1}{a}}}$, ${\varphi _k} = \frac{\pi }{a}(2k - 1)$, ${k_T} = \bmod (a,2)$,
\begin{equation}\small
\begin{aligned}
{I_a}\left( x \right) = &\frac{k_T}{a}\ln (x + 1)+ \frac{2}{a}\sum\limits_{k = 1}^{\left\lfloor {\frac{a}{2}} \right\rfloor } {\left( {\sin {\varphi _k} \cdot \arctan \left( {\frac{{x - \cos {\varphi _k}}}{{\sin {\varphi _k}}}} \right)} \right)} \\
&- \frac{1}{a}\sum\limits_{k = 1}^{\left\lfloor {\frac{a}{2}} \right\rfloor } {\left( {\cos {\varphi _k} \cdot \ln \left( {{x^2} - 2x\cos {\varphi _k} + 1} \right)} \right)} 
\end{aligned}
\end{equation}
\end{lemma}
\begin{lemma}
If $0 < p = {p_2}/{p_1} < 1,q = {q_2}/{q_1} > 1$, ${p_1},{p_2},{q_1},{q_2} \in {\bf{N^ + }}$, the settling time $T$ is estimated by $T \le \max \left\{ {{T_5},{T_6}} \right\}$, where
\begin{equation}\small
\begin{aligned}
{T_5} = \frac{{{c_T}}}{{\tau {\mu _1}\left( {q - p} \right)}}{\left( {\frac{{\tau {\mu _1}}}{{{\mu _2}}}} \right)^l} - {I_{ef}}\left( {{{\left( {\frac{{{\mu _3}}}{{(1 - \tau ){\mu _1}}}} \right)}^{\frac{1}{{{p_2}{q_1}}}}}} \right) + {I_{ef}}(0)
\end{aligned}
\end{equation}
\begin{equation}\small
\begin{aligned}
{T_6} = \frac{{{c_T}}}{{{\mu _1}\left( {q - p} \right)}}{\left( {\frac{{{\mu _1}}}{{\tau {\mu _2}}}} \right)^l} - {I_{ef}}\left( {{{\left( {\frac{{{\mu _3}}}{{(1 - \tau ){\mu _2}}}} \right)}^{\frac{1}{{{p_1}{q_2}}}}}} \right) + {I_{ef}}(0)
\end{aligned}
\end{equation}
with ${c_T} = \frac{\pi }{{\sin \left( {l\pi } \right)}}$, ${I_{ef}}(x) = \int {\frac{{{p_1}{q_1}{x^{{p_1}{q_1} - 1}}}}{{{\mu _1}{x^{{p_2}{q_1}}} + {\mu _2}{x^{{p_1}{q_2}}}}}} dx$ being an elementary function.
\end{lemma}
\subsection{Lemmas}
\begin{lemma} [\cite{lemma2}]
For any variables ${\psi _1},{\psi _2} \in {\bf{R}}$, and constants ${\gamma _1} > 0,{\gamma _2} > 0$, we have
\begin{equation}
{\left| {{\psi _1}} \right|^{{\gamma _1}}}{\left| {{\psi _2}} \right|^{{\gamma _2}}} \le \frac{{{\gamma _1}}}{{{\gamma _1}{\rm{ + }}{\gamma _2}}}{\left| {{\psi _1}} \right|^{{\gamma _1}{\rm{ + }}{\gamma _2}}}{\rm{ + }}\frac{{{\gamma _2}}}{{{\gamma _1}{\rm{ + }}{\gamma _2}}}{\left| {{\psi _2}} \right|^{{\gamma _1}{\rm{ + }}{\gamma _2}}}
\end{equation}
\end{lemma}

\begin{lemma} [\cite{lemma3}]
 For ${\chi _k} \in {\bf{R}},k = 1,2, \cdots ,s,0 < h \le 1$, one has
\begin{equation}
{\left( {\sum\limits_{k = 1}^s {\left| {{\chi _k}} \right|} } \right)^h} \le \sum\limits_{k = 1}^s {{{\left| {{\chi _k}} \right|}^h} \le {s^{1 - h}}} {\left( {\sum\limits_{k = 1}^s {\left| {{\chi _k}} \right|} } \right)^h}
\end{equation}
\end{lemma}

\begin{lemma} [\cite{2021arXiv}]
For any variable $w \in {\bf{R}}$ and constants $\delta >0, \varepsilon > 0$, it holds that
\begin{equation}
0 \le \left| w \right| - {w^2}\sqrt {\frac{{{w^2} + {\delta^2}+{\varepsilon ^2}}}{{\left( {{w^2} + {\varepsilon ^2}} \right)\left( {{w^2} + {\delta^2} } \right)}}}  < \frac{{\varepsilon \delta }}{{\sqrt {{\varepsilon ^2} + {\delta ^2}} }}  
\end{equation}
\end{lemma}

Enlightened by \emph{Lemma 7}, we can also develop the following inequalities:
\begin{lemma}
For any variable $v \in {\bf{R}}$ and positive constant $\delta_v$, we have
\begin{equation}
0 \le \left| v \right| - {v^2}\sqrt {\frac{{{v^2} + \delta _v^2}}{{{v^4} + {v^2}\delta _v^2 + \delta _v^4}}}  < 0.2576\cdot{\delta _v}
\end{equation}
\begin{equation}
0 \le \left| v \right| - \frac{2}{\pi }v \cdot \arctan \left( {\frac{v}{{{\delta _v}}}} \right) < \frac{2}{\pi }\cdot{\delta _v}
\end{equation}
\end{lemma}

\section{Main Results}
This section will provide the detailed design process of the elevation angle tracking control.
\subsection{New UBF Design}
The following PFs \cite{PF} are adopted as constraints on the tracking error ${e_1} = {x_1} - {x_{1d}}$
\begin{equation}
{K_l}\left( t \right) =
{\begin{cases}
\left( {e_1(0) - \Delta  + {e _\infty }} \right)\frac{{{T_s} - t}}{{{T_s}}}{e^{\left( {1 - \frac{{{T_s}}}{{{T_s} - t}}} \right)}} - {e _\infty },t \in \left[ {0,{T_s}} \right) \hfill \\
-{e _\infty },t \in \left[ {{T_s}, + \infty } \right)\hfill \\
\end{cases} }
\end{equation}
\begin{equation}
{K_u}\left( t \right) =
{\begin{cases}
\left( {e_1(0) + \bar \Delta  - \bar e _\infty } \right)\frac{{{T_s} - t}}{{{T_s}}}{e^{\left( {1 - \frac{{{T_s}}}{{{T_s} - t}}} \right)}} + \bar e _\infty,t \in \left[ {0,{T_s}} \right) \hfill \\
\bar e _\infty,t \in \left[ {{T_s}, + \infty } \right)\hfill \\
\end{cases} }
\end{equation}
where ${T_s} > 0,{e_\infty } > 0, \Delta  > 0, \bar e _\infty > 0, \bar \Delta  > 0$ need to be designed.

Moreover, inspired by \cite{PF}, we can also design the following novel PFs with the same properties as (21), (22):
i)
\begin{equation}
{K_{1l}}\left( t \right) =
{\begin{cases}
\frac{{\left( {{e_1}\left( 0 \right) - \Delta  + {e_\infty }} \right)}}{{{c_{1l}}}}{\mathop{\rm sech}\nolimits} \left( {\frac{{{a_{l}}t}}{{{T_s} - t}} + {b_{l}}} \right)- {e _\infty },t \in \left[ {0,{T_s}} \right) \hfill \\
-{e _\infty },t \in \left[ {{T_s}, + \infty } \right)\hfill \\
\end{cases} }
\end{equation}
\begin{equation}
{K_{1u}}\left( t \right) =
{\begin{cases}
\frac{{\left( {{e_1}\left( 0 \right) + \bar \Delta  - {{\bar e}_\infty }} \right)}}{{{c_{1u}}}}{\mathop{\rm sech}\nolimits} \left( {\frac{{{a_{u}}t}}{{{T_s} - t}} + {b_{u}}} \right) + \bar e _\infty,t \in \left[ {0,{T_s}} \right) \hfill \\
\bar e _\infty,t \in \left[ {{T_s}, + \infty } \right)\hfill \\
\end{cases} }
\end{equation}
ii)
\begin{equation}
{K_{2l}}\left( t \right) =
{\begin{cases}
\frac{{\left( {{e_1}\left( 0 \right) - \Delta  + {e_\infty }} \right)}}{{{c_{2l}}}}{\mathop{\rm csch}\nolimits} \left( {\frac{{{a_{l}}t}}{{{T_s} - t}} + {b_{l}}} \right)- {e _\infty },t \in \left[ {0,{T_s}} \right) \hfill \\
-{e _\infty },t \in \left[ {{T_s}, + \infty } \right)\hfill \\
\end{cases} }
\end{equation}
\begin{equation}
{K_{2u}}\left( t \right) =
{\begin{cases}
\frac{{\left( {{e_1}\left( 0 \right) + \bar \Delta  - {{\bar e}_\infty }} \right)}}{{{c_{2u}}}}{\mathop{\rm csch}\nolimits} \left( {\frac{{{a_{u}}t}}{{{T_s} - t}} + {b_{u}}} \right) + \bar e _\infty,t \in \left[ {0,{T_s}} \right) \hfill \\
\bar e _\infty,t \in \left[ {{T_s}, + \infty } \right)\hfill \\
\end{cases} }
\end{equation}
iii)
\begin{equation}\small
{K_{3l}}\left( t \right) =
{\begin{cases}
\frac{{\left( {{e_1}\left( 0 \right) - \Delta  + {e_\infty }} \right)}}{{{c_{3l}}}}\left[ {\coth \left( {\frac{{{a_l}t}}{{{T_s} - t}} + {b_l}} \right) - 1} \right]- {e _\infty },t \in \left[ {0,{T_s}} \right) \hfill \\
-{e _\infty },t \in \left[ {{T_s}, + \infty } \right)\hfill \\
\end{cases} }
\end{equation}
\begin{equation}\small
{K_{3u}}\left( t \right) =
{\begin{cases}
\frac{{\left( {{e_1}\left( 0 \right) + \bar \Delta  - {{\bar e}_\infty }} \right)}}{{{c_{3u}}}}\left[ {\coth \left( {\frac{{{a_l}t}}{{{T_s} - t}} + {b_l}} \right) - 1} \right] + \bar e _\infty,t \in \left[ {0,{T_s}} \right) \hfill \\
\bar e _\infty,t \in \left[ {{T_s}, + \infty } \right)\hfill \\
\end{cases} }
\end{equation}
where ${c_{1l}} = \frac{{2\exp \left( {{b_l}} \right)}}{{\exp \left( {2{b_l}} \right) + 1}},{c_{1u}} = \frac{{2\exp \left( {{b_u}} \right)}}{{\exp \left( {2{b_u}} \right) + 1}},{c_{2l}} = \frac{{2\exp \left( {{b_l}} \right)}}{{\exp \left( {2{b_l}} \right) - 1}},{c_{2u}} = \frac{{2\exp \left( {{b_u}} \right)}}{{\exp \left( {2{b_u}} \right) - 1}},{c_{3l}} = \frac{2}{{\exp \left( {2{b_l}} \right) - 1}},{c_{3u}} = \frac{2}{{\exp \left( {2{b_u}} \right) - 1}}$. $ a_l > 0, b_l > 0, a_u > 0, b_u > 0$ need to be designed.

To achieve ${K_l}\left( t \right) < {e_1} < {K_u}\left( t \right)$, a new UBF is developed as follows
\begin{equation}
{z_1} = \frac{{{c_1}}}{{{{\left( {{K_l}(t) - {e_1}} \right)}^m}}} + \frac{{{c_2}}}{{{{\left( {{K_u}(t) - {e_1}} \right)}^n}}}
\end{equation}
where $m = {m_2}/{m_1},n = {n_2}/{n_1}$, ${m_1},{m_2},{n_1},{n_2}$ are positive odd integers, and ${c_1} > 0,{c_2} > 0$ are suitable designed parameters.

It is easy to verify that when ${e_1} \to {K_l}\left( t \right),{z_1} \to  - \infty$ and ${e_1} \to {K_u}\left( t \right),{z_1} \to  + \infty$. Therefore, if $z_1$ is bounded, the tracking error can stay in the prescribed domain.

Furthermore, we construct the following new UBF, which can handle both restrained and unrestrained cases in a unified control structure.
\begin{equation}
{z_1} = \frac{{{{\left( {{K_l}(t) - {e_1}} \right)}^m}{e_1} + {c_1}}}{{2{{\left( {{K_l}(t) - {e_1}} \right)}^m}}} + \frac{{{{\left( {{K_u}(t) - {e_1}} \right)}^n}{e_1} + {c_2}}}{{2{{\left( {{K_u}(t) - {e_1}} \right)}^n}}}
\end{equation}
\begin{remark}
By selecting appropriate designed parameters, the newly developed UBFs can effectively diminish the consumption of the control input energy caused by the introduction of constraints.
\end{remark}
\subsection{Control Law Design}
Differentiating $z_1$ yields:
\begin{equation}
{\dot z_1} = {\eta _1}{\dot e_1} + {\eta _2}
\end{equation}

Then the tracking errors with compensation signals ${\zeta _1},{\zeta _2}$ are developed as
\begin{equation}
\begin{aligned}
&{w_1} = {z_1} - {\zeta _1}\\
&{w_2} = {x_2} - {x_{1,f}} - {\zeta _2}
\end{aligned}
\end{equation}
where ${x_{1,f}}$ is acquired from the following FxT differentiator \cite{2011Uniform} with the non-singular virtual control law ${\alpha _1}$ as input.
\begin{equation}
\begin{aligned}
&{{\dot x}_{1,f}} =  - {k_{1f}}{\phi _1}({\sigma _f}) + {x_{2,f}}\\
&{{\dot x}_{2,f}} =  - {k_{2f}}{\phi _2}({\sigma _f})
\end{aligned}
\end{equation}
where
\begin{equation}
\begin{aligned}
&{\phi _1}\left( {{\sigma _f}} \right) = {\mathop{\rm sig}}{\left( {{\sigma _f}} \right)^{\frac{1}{2}}} + {\mu _f}{\mathop{\rm sig}}{\left( {{\sigma _f}} \right)^{\frac{3}{2}}}\\
&{\phi _2}\left( {{\sigma _f}} \right) = \frac{1}{2}{\mathop{\rm sgn}} \left( {{\sigma _f}} \right) + 2{\mu _f}{\sigma _f} + \frac{3}{2}\mu _f^2{\mathop{\rm sig}}{\left( {{\sigma _f}} \right)^2}
\end{aligned}
\end{equation}
with ${\sigma _f} = {x_{1,f}} - {\alpha _1},{\mu _f} > 0$.

The compensation variables ${\zeta _1},{\zeta _2}$ are introduced as below [17]
\begin{equation}
\begin{aligned}
&{{\dot \zeta }_1} =  - {k_{11}}\zeta _1^p - {k_{12}}\zeta _1^q + {\eta _1}\left( {{x_{1,f}} - {\alpha _1}} \right) + {\eta _1}{\zeta _2}\\
&{{\dot \zeta }_2} =  - {k_{21}}\zeta _2^p - {k_{22}}\zeta _2^q - {\eta _1}{\zeta _1}
\end{aligned}
\end{equation}
where $1/2 < p = {p_2}/{p_1} < 1,q = {q_2}/{q_1} > 1$, ${p_1},{p_2},{q_1},{q_2}$ are positive odd integers, and ${k_{11}} > 0,{k_{12}} > 0,{k_{21}} > 0,{k_{22}} > 0$ are suitable designed parameters.

Then the novel non-singular virtual control law ${\alpha _1}$ and FxT control signal ${\bar u_1}$ are constructed as follows
\begin{equation}
\begin{aligned}
{\alpha _1} = &{{\dot x}_{1d}} - \frac{1}{{{\eta _1}}}\left( {{\eta _2}{\rm{ + }}{k_{12}}w_1^q}\right.\\
 &\phantom{=\;\;}\left.{
+{k_{11}}w_1^{1 + 2p}\sqrt {\frac{{w_1^{2 + 2p} + \delta _1^2 + \varepsilon _1^2}}{{\left( {w_1^{2 + 2p} + \varepsilon _1^2} \right)\left( {w_1^{2 + 2p} + \delta _1^2} \right)}}} } \right)\\
{{\bar u}_1} = &\frac{{{J_\alpha }}}{{{L_a}}}\left( { - {k_{21}}w_2^p - {k_{22}}w_2^q + \frac{g}{{{J_\alpha }}}{m_e}{L_a}\cos ({x_1}) }\right.\\
 &\phantom{=\;\;}\left.{+ {{\dot x}_{1,f}} - {\eta _1}{z_2} - {{\hat \Omega }_d}{w_2}\sqrt {\frac{{w_2^2 + \delta _2^2 + \varepsilon _2^2}}{{\left( {w_2^2 + \varepsilon _2^2} \right)\left( {w_2^2 + \delta _2^2} \right)}}} } \right)
\end{aligned}
\end{equation}
where ${\hat \Omega _d}$ is the modified adaptive term to compensate the disturbance, which is developed as below 
\begin{equation}
{\dot {\hat \Omega} } = {\lambda _1}\left[ {w_2^2\sqrt {\frac{{w_2^2 + \delta _2^2 + \varepsilon _2^2}}{{\left( {w_2^2 + \varepsilon _2^2} \right)\left( {w_2^2 + \delta _2^2} \right)}}}  - {\lambda _2}{{\hat \Omega }_d} - {\lambda _3}\hat \Omega _d^q} \right]
\end{equation}
where ${\lambda _1} > 0,{\lambda _2} > 0,{\lambda _3} > 0$ need to be designed.

\subsection{Stability Analysis}
\begin{theorem}
Considering the elevation channel of (1) with \emph{Assumptions 1-2}, the UBF (29) or (30), the FxT differentiator (33) with compensation mechanism (35), the adaptive compensation term (37), the virtual control law and controller (36), the attitude tracking error converges to a predefined region within a designer-defined time while the whole variables of the system are FxT bounded. 
\end{theorem}
\begin{proof}
The Lyapunov function is constructed as below
\begin{equation}
V = \frac{1}{2} \lambda _1^{ - 1}{\left( {{{\hat \Omega }_d} - {{\bar d}_1}} \right)^2} + \frac{1}{2}\sum\limits_{i = 1}^2 {\left( {w_i^2 + \zeta _i^2} \right)} 
\end{equation}

Then when $t > {t_1}$, the derivative of $V$ with respect to time is calculated as
\begin{equation}
\begin{aligned}
\dot V \le & - \left[ {\left( {\frac{{{k_{11}}}}{{1 + p}} - \frac{1}{2}} \right)\zeta _1^{1 + p} + \frac{{{k_{21}}}}{{1 + p}}\zeta _2^{1 + p} + \sum\limits_{i = 1}^2 {\frac{{{k_{i1}}p}}{{1 + p}}w_i^{1 + p}} } \right] \\
&- \left[ {\left( {\frac{{{k_{12}}}}{{1 + q}} - \frac{1}{2}} \right)\zeta _1^{1 + q} + \frac{{{k_{22}}}}{{1 + q}}\zeta _2^{1 + q} + \sum\limits_{i = 1}^2 {\frac{{{k_{i2}}q}}{{1 + q}}w_i^{1 + q}} } \right] \\
&-\left[ {\frac{{{\lambda _2}}}{2}{\left( {{{\hat \Omega }_d} - {{\bar d}_1}} \right)^{1 + p}} + \frac{q}{{1 + q}}{{\left( {{{\hat \Omega }_d} - {{\bar d}_1}} \right)}^{1 + q}}} \right] + {\mu _3}
\end{aligned}
\end{equation}

With the aid of \emph{Lemma 6}, it follows from (39) that
\begin{equation}
\dot V(x) \le  - {\mu _1}V{(x)^{\frac{{1 + p}}{2}}} - {\mu _2}V{(x)^{\frac{{1 + q}}{2}}}{\rm{ + }}{\mu _3}
\end{equation}
where ${\mu _1} > 0,{\mu _2} > 0,{\mu _3} > 0$. 

In light of \emph{Lemma 1}, the closed-loop system is practical FxT stable, and ${w_i},{\zeta _i}$ converge to the following residual set within a fixed time:
\begin{equation}
\begin{aligned}
&\left| {{w_i}} \right| \le \min \left\{ {\sqrt {2{{\left( {\frac{{{\mu _3}}}{{\left( {1 - \tau } \right){\mu _1}}}} \right)}^{\frac{2}{{1 + p}}}}} ,\sqrt {2{{\left( {\frac{{{\mu _3}}}{{\left( {1 - \tau } \right){\mu _2}}}} \right)}^{\frac{2}{{1 + q}}}}} } \right\}\\
&\left| {{\zeta _i}} \right| \le \min \left\{ {\sqrt {2{{\left( {\frac{{{\mu _3}}}{{\left( {1 - \tau } \right){\mu _1}}}} \right)}^{\frac{2}{{1 + p}}}}} ,\sqrt {2{{\left( {\frac{{{\mu _3}}}{{\left( {1 - \tau } \right){\mu _2}}}} \right)}^{\frac{2}{{1 + q}}}}} } \right\}
\end{aligned}
\end{equation}

Then $\left| {{z_1}} \right| \le \left| {{w_1}} \right| + \left| {{\zeta _1}} \right|$ is bounded. According to the property of UBF (29) and (30), it is concluded that the tracking error can converge to a predefined domain within a designer-defined time.

This completes the proof.
\end{proof}

\section{Simulation Results}
To verify the feasibility and effectiveness of the developed control protocol, numerical simulations are performed on the elevation channel. The state is initialized at ${x_1}(0) =  - \frac{{2\pi }}{{15}}(rad)$ and the disturbance is ${d_1} = 0.3 \sin (2t)({{rad} \mathord{\left/
 {\vphantom {{rad} {{s^2}}}} \right.
 \kern-\nulldelimiterspace} {{s^2}}})$. The given command is set as:
\begin{equation}
{x_{1d}}(t){\rm{ = }}\frac{\pi }{{18}}\sin (0.3\pi t - \frac{\pi }{2})
\end{equation}
\subsection{The effectiveness verification of the designed UBF}
In this simulation study, the PF parameters are designed as: ${T_s} = 1.2,\Delta  = 0.05,{e_\infty } = 0.03$, and the parameters of the developed control protocol are presented as follows: ${k_{11}} = 1,{k_{12}} = 1.5,{k_{21}} = 2,{k_{22}} = 6,{c_1} = {c_2} = 1,{\delta _1} = 0.1,p = \frac{3}{5},q = \frac{5}{3}$, while $m = n = 1,\frac{1}{3},\frac{1}{5},\frac{1}{7}$ respectively to show the performance of the designed UBF. The UBF in \cite{UBF} is also utilized as the contrasting method, which is marked as classical-UBF.
\begin{figure}
	\centering
	\subfigure[Control inputs with classical and designed UBF]{
	\includegraphics[width=0.4\textwidth]{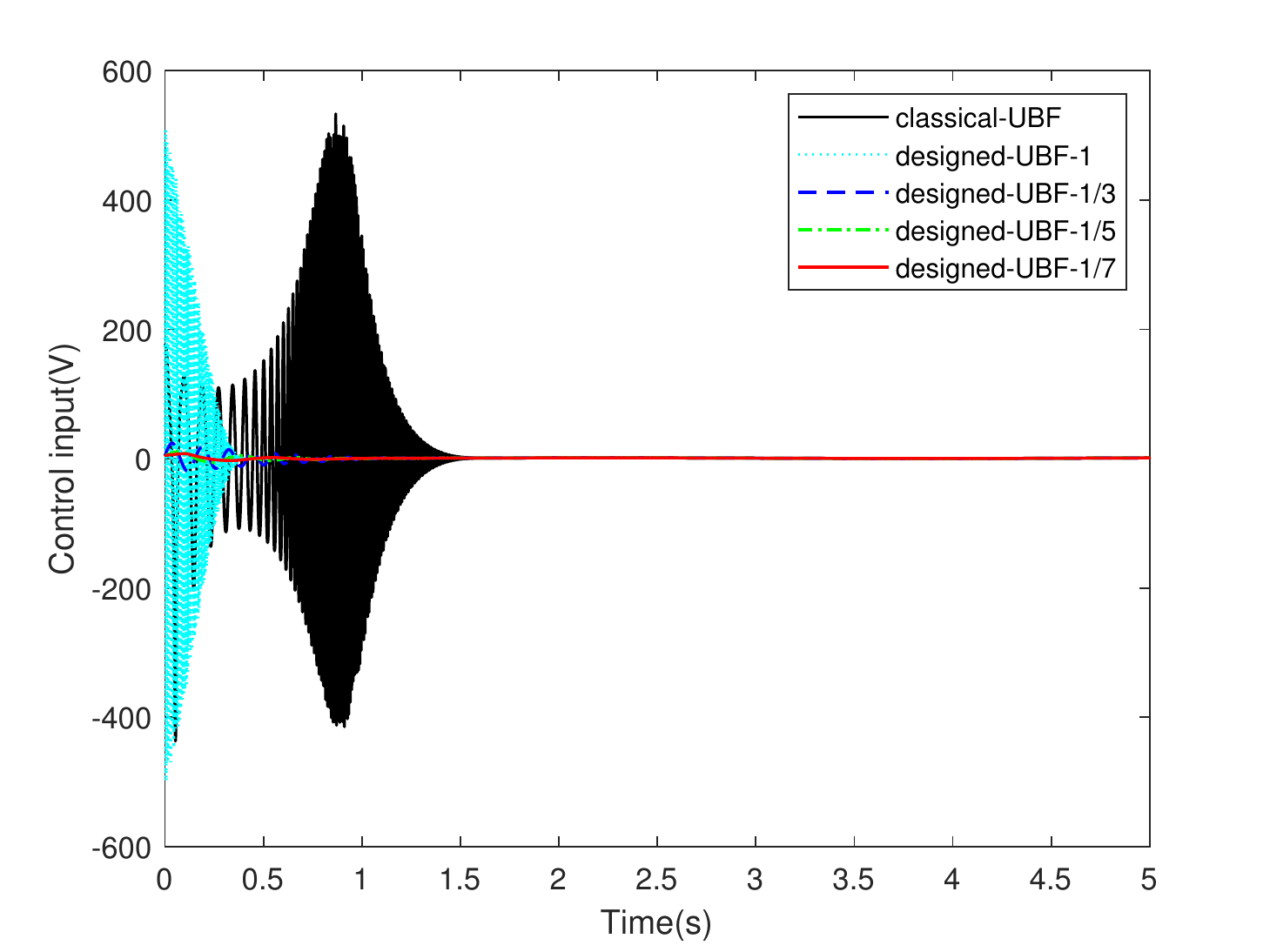}}
	\subfigure[Control inputs with designed UBF]{
	\includegraphics[width=0.4\textwidth]{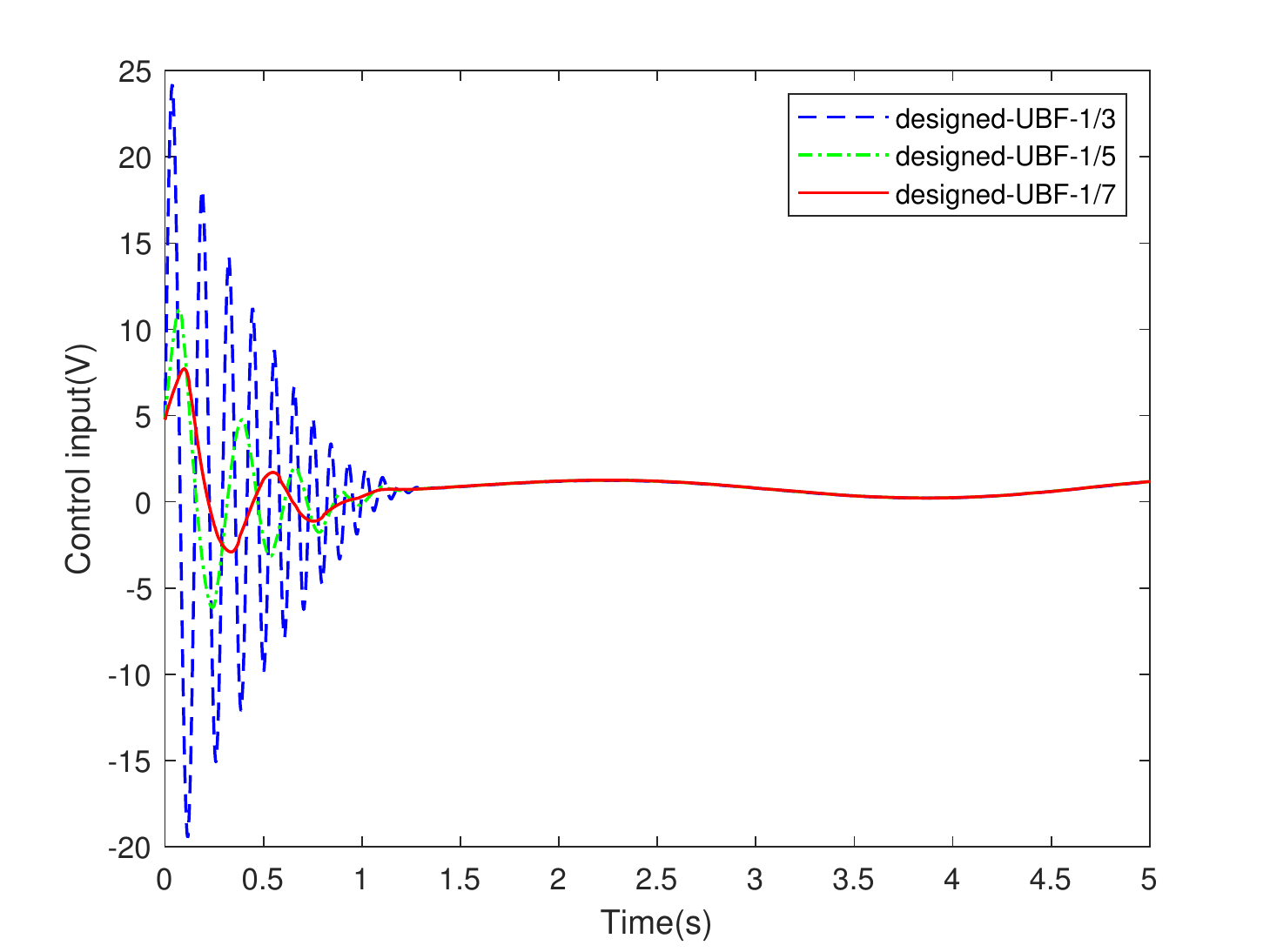}}
	\caption{Control inputs}
\end{figure}

The numerical simulation results are plotted in Fig. 2 (a)-(b). It is observed that by selecting appropriate designed parameters, the newly developed UBF can effectively diminish the consumption of the control input energy caused by the introduction of constraints.
\subsection{The simulation of attitude tracking control}
In this study, the PF parameters are designed as: ${T_s} = 1.2,\Delta  = 0.05,{e_\infty } = 0.01$, and the parameters of the designed UBF are selected as: ${c_1} = {c_2} = 0.2,m = n = \frac{1}{7}$, while the other parameters are identical to those in simulation I. The developed control protocol without PF, and the CFB algorithm in \cite{2.li2021} without RBFNN are adopted as the contrasting approaches.
\begin{figure}\centering
	\includegraphics[width=0.4\textwidth]{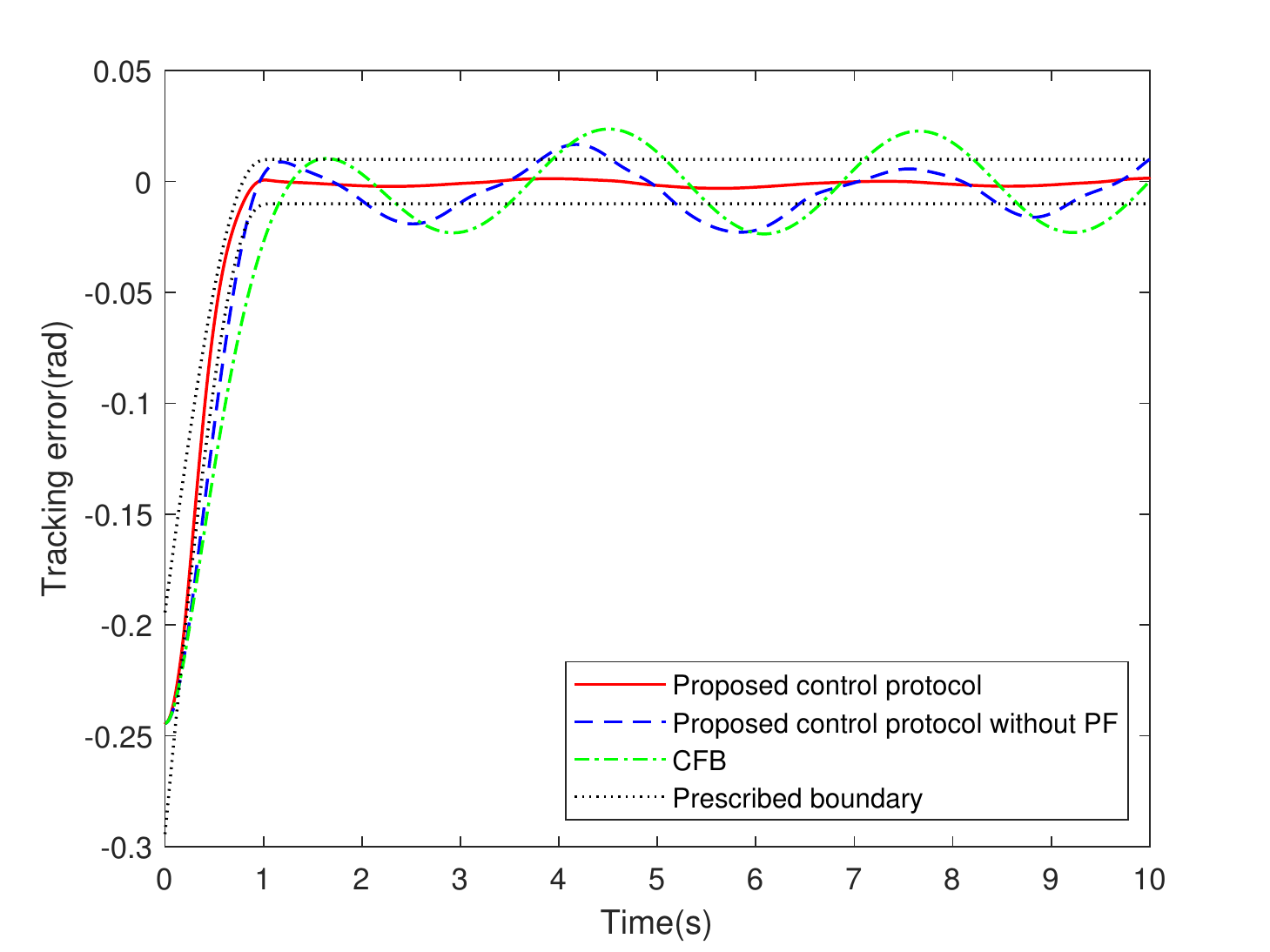}
	\caption{Tracking errors}
\end{figure}

The curves of the tracking error are displayed in Fig. 3, which demonstrates that the attitude tracking error can converge to a predefined region within a designer-defined time. Furthermore, the tracking error owns better transient behavior, especially small overshoot.

\section{Conclusion}
A novel prescribed performance adaptive fixed-time backstepping control protocol was proposed in this article to tackle the attitude tracking problem of a 3-DOF helicopter subject to disturbances. A newly designed UBF was adopted to transform the prescribed performance constrained system into an unconstrained one. By virtual of an improved FxT backstepping control algorithm, the attitude tracking was fulfilled. A modified FxT stability theorem was presented to provide a less conservative and more accurate approximation of the settling time than the classical result. Theoretical analysis proves that all signals of the system are FxT bounded, while the tracking error converges to a predetermined domain with small overshoot in a user-defined time. Simulation results show the feasibility and effectiveness of the presented control strategy.

\balance
\bibliographystyle{Bibliography/IEEEtranTIE}
\bibliography{Bibliography/IEEEabrv,Bibliography/myRef}\ 

\end{document}